\documentclass[11pt]{amsart}
\usepackage{amsmath, amssymb, amsthm, url}

\usepackage{graphicx}
\allowdisplaybreaks

\usepackage{graphicx}

\usepackage{amsmath}
\usepackage[mathscr]{eucal}
\usepackage{amscd}
\usepackage{amsfonts,enumerate,epsfig,fancyhdr,subfigure}
\usepackage{mathrsfs}
\usepackage{multirow}
\usepackage{multicol}
\usepackage{pdflscape} %landscape - portrait
\usepackage{longtable}
\usepackage[utf8]{inputenc}
\usepackage{color}
\usepackage{mathtools}

\usepackage{xcolor,colortbl}   %for table coloring
\def\cc{\cellcolor}

\theoremstyle{plain}
\newtheorem{theorem}{Theorem}[section]
\newtheorem{proposition}[theorem]{Proposition}
\newtheorem{lemma}[theorem]{Lemma}
\newtheorem{corollary}[theorem]{Corollary}

\theoremstyle{definition}

\newtheorem{example}[theorem]{Example}

\def\x{{\mathbf x}}
\def\y{{\mathbf y}}
\def\cc{{\mathbf c}}

\def\0{{\mathbf 0}}
\def\1{{\mathbf 1}}

\def\g{{\mathbf g}}

\def\s{{\sigma}}

\def\g{{\gamma}}

\def\M{{\mathcal M}}

\def\D{{\mathbb D}}
\def\0{{\mathbf 0}}
\def\1{{\mathbf 1}}

\def\CC{{\mathcal C}}

\def\x{{\mathbf x}}
\def\y{{\mathbf y}}

\def\0{{\mathbf 0}}
\def\1{{\mathbf 1}}

\def\g{{\mathbf g}}

\def\s{{\sigma}}

\def\g{{\gamma}}

\def\M{{\mathcal M}}

\def\D{{\mathbb D}}
\def\0{{\mathbf 0}}
\def\1{{\mathbf 1}}

\def\CC{{\mathcal C}}

\newenvironment{smatrix}{\left(\begin{smallmatrix}}{\end{smallmatrix}\right)}

\begin{document}

	\title[ An improved upper bound on self-dual codes over finite fields $GF(11), GF(19)$, and $GF(23)$ ]
	{An improved upper bound on self-dual codes over finite fields $GF(11), GF(19)$, and $GF(23)$ }

	\author{Whan-Hyuk Choi${}^{*}$ \and Jon-Lark Kim${}^{**\dag}$   }

	\address{Department of Mathematics, Sogang University, Seoul 04107, Republic of Korea, \ Email: whchoi@kangwon.ac.kr}
	
	\address{Department of Mathematics, Sogang University, Seoul 04107, Republic of Korea, \ Email: jlkim@sogang.ac.kr }
	
	\date{}
	

	\thanks{\tiny ${}^{\dag}$Corresponding author. \\
		The author${}^{*}$ is supported by the
		National Research Foundation of Korea (NRF) grant funded by the
		Korea government (NRF-2019R1I1A1A01057755).
		The author${}^{**}$ is supported by the
		National Research Foundation of Korea (NRF) grant funded by the
		Korea government (NRF-2019R1A2C1088676).}

	%%%%%%%%%%%%%%%%%%%%%%%%%%%%%%%%%%%%%%%%%%%%%%%%

	\maketitle

\begin{abstract}
	This paper gives new methods of constructing {\it symmetric self-dual codes} over a finite field $GF(q)$ where $q$ is a power of an odd prime. These methods are motivated by the well-known Pless symmetry codes and quadratic double circulant codes. Using these methods, we construct an amount of symmetric self-dual codes over $GF(11)$, $GF(19)$, and $GF(23)$ of every length less than 42. We also find 153 {\it new} self-dual codes up to equivalence: they are $[32, 16, 12]$, $[36, 18, 13]$, and $[40, 20,14]$ codes over $GF(11)$, $[36, 18, 14]$ and $[40, 20, 15]$ codes over $GF(19)$, and $[32, 16, 12]$, $[36, 18, 14]$, and $[40, 20, 15]$ codes over $GF(23)$. They all have new parameters with respect to self-dual codes. Consequently, we improve bounds on the highest minimum distance of self-dual codes, which have not been significantly updated for almost two decades.	
\end{abstract}
	
	%%%%%%%%%%%%%%%%%%%%%%%%%%%%%%%%%%%%%%%%%%%%%%%%
	
	\section{Introduction}
	Coding theory, one of the most interesting areas of applied mathematics, was born almost simultaneously with the invention of modern computers - the beginning of the error-correcting code came from Claude Shannon's paper ``A mathematical theory of communication'' in 1948, and Richard W. Hamming's paper ``Error detecting and error correcting codes'' in 1950. These days, binary and nonbinary codes such as $q$-ary Hamming codes, the binary and ternary Golay codes, and $q$-ary Reed-Solomon codes are used in internet communication, GPS signals, mobile phones, and computer devices.
	It is well known that error-correcting codes are closely related to cryptography \cite{calkavur2020},~\cite{Sendrier2017}.
	Moreover, researchers have recently started investigating the relation between error-correcting codes and deep learning.~\cite {BeEry2019},~\cite{Huang2019}.
	
	On the other hand, self-dual codes have been the subject of much interest and are regarded as one of the most important classes of error-correcting codes. This is because of both theoretical reason and connections to various fields of mathematics such as designs \cite{Harada2017}, lattices \cite{Bannai1999}, sphere-packings \cite{Conway1999}, and modular forms \cite{bernhard1996}.
	
	Among various research topics of self-dual codes, it has attained an extensive research effort to find a {\it best} code; here, {\it best} refers to having the greatest error correction ability as possible. The error correction capability of a code depends on the minimum distance. Thus, it is crucial to find a method to construct codes having the highest minimum distance. To this end, various techniques are studied involving circulant and bordered circulant matrices \cite{betsumiya2003},~\cite{grassl2009} and quadratic double circulant matrices \cite{Gaborit2002}. Recently, families of codes over rings have been used to construct self-dual codes over finite fields \cite{Doughert2020},~\cite{kim2018}.
	
	Despite these efforts, there remain many codes to be found, missed by previous construction methods due to computation complexity. In particular, we hardly know about the optimal minimum distances of self-dual codes over finite fields of order $\ge$ 5 and of lengths $\ge$ 22. % as shown in Table \ref{previous_results}.
	In this case, only the possible bounds of highest minimum distances are known so far. For example, in the case of codes over $GF(11)$, the bounds of highest minimum distances of lengths $\le$ 40 are known, as we can see in Table \ref{previous_results_11}. Moreover, there is no information about the lower bound of the self-dual code of length 28.%We revised these results in Table \ref{previous_results}.
	
	In 1972, Vera Pless introduced {\it Pless symmetry codes}, as a generalization of ternary extended Golay code \cite{VP1},~\cite{VP2}. Using this class of codes, Pless obtained many new optimal self-dual codes over $GF(3)$. Three decades later, Gaborit presented a generalization of Pless symmetry codes to different fields, {\it quadratic double circulant codes}~\cite{Gaborit2002}. He also found many new self-dual codes over $GF(4)$, $GF(5)$, $GF(7)$, and $GF(9)$. We want to remark two things: one is that these two methods used particular symmetric matrices to construct self-dual codes. The other is that these methods have a limitation of lengths; the possible lengths of codes are limited to $n+1$ or $n-1$ where $n$ is a power of an odd prime. Thus, there needs a new method to fill the gap between these lengths. These are the main motivation of this paper.
	
	In \cite{choi2020}, we introduced a method of {\it symmetric building-up construction}. If a self-dual code has a symmetric generator matrix, it is called {\it a symmetric self-dual code}. This method was to construct symmetric self-dual codes over $GF(q)$ for $q \equiv 1 \pmod 4$. In \cite{choi2020}, we showed that this method provides an efficient way to construct all symmetric self-dual codes over $GF(q)$, increasing lengths by two.
	Stimulated by this result, we have struggled to find a method when $q \equiv 3 \pmod 4$. However, it is not easy to generalize the method in \cite{choi2020}. In \cite{choi2020}, the square root of -1 plays the key role, but unfortunately, it is well-known that the square root of -1 does not exist in $GF(q)$ for $q \equiv 3 \pmod 4$.
	Nevertheless, we find two novel construction methods as follows :
	\begin{enumerate}
		{\bf \item[1.] Construction A}\\
		Let
		$(I_n \mid A)$
		be a generator matrix of a symmetric self-dual code of length $2n$ over $GF(q)$ and assume that $(\x_n , \y_n)$ is a codeword satisfying $\x_n \cdot \y_n = 0$ and $\x_n \cdot \x_n=k$ such that $-1\pm k$ are squares in $GF(q)$. And let $B= \left( \begin{matrix}
		\alpha \x_n + \beta \y_n \\
		\beta \x_n - \alpha \y_n 					
		\end{matrix}\right)$ where $\alpha^2+\beta^2 =-1$,
		$E= \frac{1}{k} ( s \x_n^T \x_n  + t \y_n^T \y_n -  \x_n^T \y_n - \y_n^T \x_n)$ where $s^2 = -1+k$ and $t^2=-1-k$ and let $D = -\frac{1}{k^2}B(A+E_1)B^T B B^T$. Then
		$$
		\left(\begin{array}{c|c|c|c}
		I_2 & O& D & B \\
		\hline
		O& I_n&B^T& A+E\\
		\end{array}\right)$$
		is a generator matrix of symmetric self-dual code of length $2n+4$.
		
		{\bf 	\item[2.] Construction B}\\
		Let
		$(I_n \mid A)$
		be a generator matrix of a symmetric self-dual code of length $2n$ over $GF(q)$, let $P=\begin{smatrix}
		\alpha & \beta \\ \beta & -\alpha
		\end{smatrix}
		$ be a $2 \times 2$ matrix such that $P^2=-I_2$, and let a matrix $M = \left( \begin{matrix}
		\x\\\beta^{-1} \x (A - \alpha I)
		\end{matrix}\right)$ for a vector $\x$ in $GF(q)^n$.
		Assume that $H$ is a $2\times 2$ symmetric matrix satisfying $(H+P)(H-P)=-M M^T$ and $H-P$ is non-singular. Then
		
		\qquad \qquad \qquad $
		\left(\begin{array}{c|c|c|c}
		I_2 & O& H & M\\
		\hline
		
		O& I_n&M^T& A+M^T (H-P)^{-1} M\\
	\end{array}\right)
	$
	
	is a generator matrix of symmetric self-dual code of length $2n+4$.
\end{enumerate}

Using these methods, we obtain many new self-dual codes. Consequently, we improve the bounds on the minimum distances of self-dual codes. We revised these results in Table \ref{our_results1}. In Table  \ref{our_results1}, new parameters are written in bold. Throughout this paper, $d_{sym}$ denotes the highest minimum distance of a symmetric self-dual code over $GF(p)$ and $d_{sd}$ denotes the previously best-known minimum distance of self-dual codes over $GF(p)$. More precisely,
we give {\it new} self-dual codes with highest minimum weights: they are $[32, 16, 12]$, $[36, 18, 13]$, and $[40, 20,14]$ codes over $GF(11)$, $[36, 18, 14]$ and $[40, 20, 15]$ codes over $GF(19)$, and $[32, 16, 12]$, $[36, 18, 14]$, and $[40, 20, 15]$ codes over $GF(23)$. We also provide numbers of new symmetric self-dual codes, up to equivalence, in Table \ref{numbers}.
%We checked the equivalence with the computer algebra system \textsc{Magma} \cite{Magma}.

\begin{table}[h]
\begin{center}
	\begin{small}
		\begin{tabular}{|c|r|r|r|r|r|r|}
			\hline
			$p$ &      \multicolumn{2}{c|}{$11$}   &	      \multicolumn{2}{c}{$19$}   &    \multicolumn{2}{|c|}{$23$}     \\
			\hline
			$n$
			&\multicolumn{1}{c|}{$d_{sym}$} 	& \multicolumn{1}{c|}{$d_{sd}$} &  \multicolumn{1}{c|}{$d_{sym}$} 	& \multicolumn{1}{c|}{$d_{sd}$} & \multicolumn{1}{c|}{$d_{sym}$} 	& \multicolumn{1}{c|}{$d_{sd}$}	  \\
			\hline
			4  &{\it 3} 	& 3 	&{\it 3}			&3 &{\it 3} 			&3 	\\\hline
			8  &{\it 5} 	& 5		&{\it 5}			&5  	&{\it 5} 			&5  		\\\hline
			12  &{\it 7} 	& 7 	&{\it 7}			&7  	&{\it 7} 			&7 		\\\hline
			16  &{\it 7} 	&9  	&{\it 8}			&8 	&{\it 8}			&9   			\\\hline
			20  &{\it 8} 	& 10 		&{\it 11}			&11  &{\it 9}	 		&10 		\\\hline
			24  &{\it 9} 	& 9 	&{\it 10}			&10  	&{\it 10} 		&13	 	\\ \hline
			28  &{\it 10}   & 10 			&{\it 11}			&11	   	&{\it 11}			&11 		\\\hline
			32  &\textbf{\textit{12}}   &?							&{\it 12}   		&14		&\textbf{\textit{12}} 		&?		\\\hline
			36  &\textbf{\textit{13}}   &$12$		&\textbf{\textit{14}}		&?		  				&\textbf{\textit{14}} 		&$12$	\\\hline
			40  &\textbf{\textit{14}}   & $13$	&\textbf{\textit{15}}		&?		  				&\textbf{\textit{15}} 		&$13$		\\\hline
			
		\end{tabular}
		\caption{The highest minimum distance $d_{sym}$ of symmetric self-dual codes vs. previously best known minimum distance $d_{sd}$ of self-dual codes \cite{betsumiya2003, choi2020, deboer1996, Gaborit2002, grassl2008, Gulliver2008, Shi2018}. New parameters are written in bold.}
		\label{our_results1}
	\end{small}
\end{center}	
\end{table}

\begin{table}
\begin{center}
	\begin{small}
		\begin{tabular}{|c|c|c|c|c|c|c|}
			\hline
			$p$ &      \multicolumn{2}{c|}{$11$}   &	      \multicolumn{2}{c}{$19$}   &    \multicolumn{2}{|c|}{$23$}     \\
			\hline
			$n$
			& $d_{sym}$         & \# of codes     &   $d_{sym}$         &  \# of codes &  $d_{sym}$   &  \# of codes  \\
			\hline
			32  &12			&$\ge 44$	&12 & $\ge 801$		&12			&$\ge 52$	\\\hline
			36 &13			&$\ge 16$    &14 &   $\ge 3$    &14			&$\ge 2$      	\\\hline
			40  &14			&$\ge 42$		&15 &$\ge 2$		&15			&$\ge 1$	\\\hline
			
		\end{tabular}
		\caption{Numbers of new symmetric self-dual code of length 32, 36 and 40}
		\label{numbers}
	\end{small}
\end{center}	
\end{table}

The paper is organized as follows. Section 2 gives preliminaries for self-dual codes over finite fields.
In Section 3, we present two construction methods for {\it symmetric self-dual codes} over $GF(q)$, where $q$ is an odd prime power. In Section 4, we give the improved bounds of highest minimum distances and the computational results of the best codes obtained using our new methods.
All computations in this paper were done with the computer algebra system \textsc{Magma} \cite{Magma}.

We use the following notations throughout this paper.

\begin{tabular}{  c  p{9cm} }

\textbf{Notations}& \\
\hline

$q$  & a power of an odd prime number \\
$GF(q)$ & finite field of order $q$\\
$d_{sym}$ & the highest minimum distance of symmetric self-dual codes \\
$d_{sd}$ & the previous best known minimum distance of self-dual codes \\
$I_n$ & the identity matrix of degree $n$ \\
$[n,k,d]_q$ code& a linear code of length $n$ and dimension $k$ over $GF(q)$ with minimum distance $d$\\
$A^{-1}$ & the inverse of a matrix $A$ \\
$A^T$ & the transpose of a matrix $A$ \\

\hline
\end{tabular}

\section{Preliminaries}

Let $n$ be a natural number, and $GF(q)$ be the finite field of order $q$ where $q$ is a prime power. A {\it linear code} $\CC$ of length $n$ and dimension $k$ over $GF(q)$ is a $k$-dimensional subspace of $GF(q)^n$. An element of $\CC$ is called a {\it codeword}. A {\it generator matrix} of $\CC$ is a matrix whose rows form a basis of $\CC$; therefore, a generator matrix of a linear code $\CC$ of length $n$ and dimension $k$ over $GF(q)$ is a $k \times n$ matrix over $GF(q)$.
For vectors $\x = (x_i )$ and $\y=(y_i)$ in $GF(q)^n$, we define the inner product $\x \cdot \y = \sum_{i=1}^{n} x_i y_i$. If vectors are identified with row matrices, the  inner product can also be written as a matrix multiplication $\x \cdot \y =\x \y ^T$, where $\y^T$ denotes the transpose of $\y$. For a linear code $\CC$, {\it dual code $\CC^{\perp}$} is defined as a set of orthogonal vectors of $\CC$, i.e., $$\CC^{\perp}=\{\x\in GF(q)^n \mid \x \cdot \cc =0 \text{ for all $\cc\in C$} \}.$$  A linear code $\CC$ is called {\it self-dual} if $\CC = \CC^{\perp}$ and {\it self-orthogonal} if $\CC \subset \CC^{\perp}$.

The {\it weight} of a codeword $\cc$ is the number of non-zero symbols in the codeword and denoted by $wt(\cc)$. The {\it Hamming distance} between two codewords $\x$ and $\y$ is defined by $d(\x,\y)=wt(\x-\y)$. The {\it minimum distance} of $\CC$, denoted by $d(\CC)$, is the smallest Hamming distance between distinct codewords in $\CC$. The minimum distance determines the error-capability; thus, the minimum distance is regarded as the most important parameter of a code. If a code has the minimum distance that meets some upper bounds, it is called an {\it optimal code}.
It is well-known \cite[chapter 2.4.]{HP3} that a linear code of length $n$ and dimension $k$ satisfy the Singleton bound, $$d(\CC) \le n - k +1.$$ A code that achieves the equality in the Singleton bound is called a \textit{maximum distance separable(MDS)} code. Obviously, a self-dual code of length $2n$ over $GF(q)$ is MDS if the minimum distance equals $n+1$. Although every MDS code is optimal, the MDS conjecture shows that there exists an MDS self-dual code of length $2n$ over $GF(q)$ only if $2n \le q+1$ \cite{Ball2012}. Therefore, if $2n > q+1$, the  minimum distance of self-dual code of length $2n$ over $GF(q)$ is upper bounded by $n$.

Let $I_n$ be a identity matrix of order $n$ and let $A^T$ denote the transpose of a matrix $A$. It is well-known that a self-dual code $\CC$ of length $2n$ over $GF(q)$ is equivalent to a code with a standard generator matrix
\begin{equation}\label{std-form}
\left(
\begin{array}{c|c}
	I_{n}& A
\end{array}
\right),
\end{equation} where $A$ is a $n \times n$ matrix satisfying $AA^T=-I_{n}$.

A matrix $A$ is called {\it symmetric} if $A^T=A$. If a self-dual code of length $2n$ over $GF(q)$ has a standard generator matrix $G=(I_n \mid A)$ where $A$ is symmetric, it is called {\it a symmetric self-dual code}.
Since the class of symmetric self-dual codes is a subclass of general self-dual codes, the bound on minimum distances of symmetric self-dual code may be different from that of self-dual codes. However, if a symmetric self-dual code has the same parameter as an optimal(resp. MDS) self-dual code, it is called a {\it optimal (resp. MDS) symmetric self-dual code}. If the minimum distance of a symmetric self-dual code meets the best known minimum distance of a self-dual code, it is called a {\it best symmetric self-dual code}.

In \cite{VP1}, Pless introduced {\it Pless symmetry codes} as a generalization of ternary extended Golay code and their construction method. As a result, Pless obtained optimal self-dual codes of length 24, 36, 48, and 60 over $GF(3)$. Later in \cite{Gaborit2002}, Gaborit presented a generalization of Pless symmetry codes to different fields, {\it quadratic double circulant codes} and their \linebreak construction method. Gaborit obtained many new self-dual codes over $GF(4)$, $GF(5)$, $GF(7)$ and $GF(9)$, and improved the bounds on the highest minimum \linebreak distances.
To use as a reference, we additionally obtain quadratic double circulant codes of lengths $\le 40$ over various finite fields, following the same construction method in \cite{Gaborit2002}. We present these codes in Table \ref{qdcodes}, following the same notations in \cite{Gaborit2002}.

\begin{table}[h]
\begin{center}
	\begin{small}
		\begin{tabular}{crrr|crrr}
			\hline
			length &$q$ & generator matrix &$d$ & length &$q$ & generator matrix &$d$  \\
			\hline

			%28 & 7  & $\mathscr{S}_{13}(1,0)$ &  10  &  \\
			28 & 11  & $\mathscr{S}_{13}(3,0)$ &  10 &   28 & 17  & $\mathscr{S}_{13}(2,0)$ &  10  \\
			28 & 19  & $\mathscr{S}_{13}(5,0)$ & 10  &  28 & 29  & $\mathscr{S}_{13}(4,0)$ &  10  \\

			%36 &  3  & $\mathscr{S}_{17}(1,0)$ & 12  &  \\
			%36 & 	7  & $\mathscr{S}_{17}(2,0)$ &  12  &  \\
			\hline
			36 & 11  & $\mathscr{S}_{17}(4,0)$ &  12 &  36 & 13  & $\mathscr{S}_{17}(3,0)$ &  12   \\
			36 & 17  & $\mathscr{S}_{17}(7,0)$ & 12  &  36 & 23  & $\mathscr{S}_{17}(11,0)$ &  12   \\
			
			%40 &  3  & $\mathscr{S}_{19}(1,1)$ & 12  &  \\
			%	40 & 	5  & $\mathscr{S}_{19}(1,0)$ &  13  &  \\
			%	40 & 	7  & $\mathscr{S}_{19}(3,0)$ &  13  &				
			\hline
			40 & 11  & $\mathscr{S}_{19}(3,4)$ &  13   &	40 & 13  & $\mathscr{S}_{19}(2,4)$ &  13  \\
			40 & 17  & $\mathscr{S}_{19}(7,0)$ & 13   &	40 & 23  & $\mathscr{S}_{19}(2,0)$ &  13  \\
			40 & 29  & $\mathscr{S}_{19}(1,3)$ & 13  &&& & \\
			
			%48 & 5  & $\mathscr{S}_{23}(1,1)$ &  14 		&  \\
			%	48 & 7 & $\mathscr{S}_{23}(1,2)$ &  15   & \\
			%	48 & 11  & $\mathscr{S}_{23}(1,3)$ &  15    	&  \\
			
			\hline
		\end{tabular}
	\end{small}
	\caption{New quadratic double circulant codes over $GF(q)$ obtained using methods in \cite{Gaborit2002}}
	\label{qdcodes}
\end{center}
\end{table}

We remark that a self-dual code in the class of Pless symmetry codes or quadratic double circulant codes is equivalent to a symmetric self-dual code.  In general, a pure double circulant self-dual code is equivalent to a symmetric self-dual code, and a bordered double circulant self-dual code is equivalent to a symmetric self-dual code under a certain condition. We discuss the equivalence between these codes in the next.

Let $S_n$ be a symmetric group of order $n$ and $\D^n$ be the set of diagonal matrices over $GF(q)$ of order $n$, $$ \D^n=\{diag(\g_i) \mid  \g_i\in GF(q), \g_i^2=1\}.$$ The group of all  \textit{$\gamma$-monomial transformations of length $n$},  $\M^n$  is defined by
$$\M^n=\{p_{\s} \gamma \mid \gamma \in\D^n, \s \in S_n\}
$$ where $p_{\sigma}$ is the permutation matrix corresponding $\sigma \in S_n$. We note that a $\gamma$-monomial transformation preserves the self-orthogonality of a code (see \cite[Thm 1.7.6]{HP3}).
Let $\CC\tau =\{\cc\tau \mid \cc\in \CC\}$ for an element $\tau$ in $\M^{n}$ and a code $\CC$ of length $n$. If there exists an element $\mu \in \M^{n}$ such that $\CC \mu=\CC'$ for two distinct codes $\CC$ and $\CC'$, then $\CC$ and $\CC'$ are called \textit{equivalent} and denoted by $\CC \simeq \CC'$ .% An {\it automorphism} of $\CC$ is an element $\mu \in \M^{2n}$ satisfying $\CC \mu=\CC$. The set of all automorphisms of $\CC$ forms the {\it automorphism group} $\Aut(\CC)$ as a subgroup of $\M^{2n}$.

\begin{proposition}\label{M-equivalent}
Let $G=(    I_n \mid  A )$  and $G'=(    I_n \mid  B )$ be generator matrices of self-dual codes $\CC$ and $\CC'$ of length $2n$, respectively. If $A = \mu_1 B \mu_2$ for some $\mu_1, \mu_2 \in \M^n$, then $\CC$ and $\CC'$ are equivalent.

\begin{proof}
	%Let    $\tau=(1,n+1)(2,n+2)\cdots(k,n+k)\cdots(n,2n)\in S_{2n}$. Then,
	For $\mu =\left(\begin{array}{c|c}
	\mu_1^{-1} &O \\
	\hline
	O & \mu_2			\end{array}\right) \in \M^{2n}$,
	$$(     I_n \mid  A )
	=(    I_n   \mid  \mu_1 B \mu_2)=(    \mu_1^{-1}   \mid  B \mu_2)
	= (      I_n \mid  B ) \mu. $$ Thus, $\CC$ and $\CC'$ are equivalent.
\end{proof}
\end{proposition}

\begin{corollary}\label{cor2.6}Let $I_n$ be the identity matrix of order $n$, $A$ is an $n\times n$ circulant matrix, $B$ is an $(n-1)\times (n-1)$ circulant matrix. Then,
\begin{enumerate}
	\item a pure double circulant code over $GF(q)$ with a generator matrix of the form $$(I_n \mid A)$$ is equivalent to a code with symmetric generator matrix, and
	
	\item a bordered double circulant code  over $GF(q)$ with a generator matrix of the form \[
	\left(
	\begin{array}{ccc}
	& \alpha & \beta \cdots \beta \\
	\raisebox{-10pt}{{\large\mbox{{$I_n$}}}} & \gamma \beta& \raisebox{-15pt}{{\large\mbox{{$A$}}}} \\[-4ex]
	& \vdots & \\[-0.5ex]
	& \gamma  \beta &
	\end{array}
	\right),
	\] where $\alpha$ and $\beta$ are elements in $GF(q)$ and $\gamma^2=1$, is equivalent to a code with symmetric generator matrix.
\end{enumerate}

\end{corollary}
\begin{proof}
It is clear that a column reversed matrix of a circulant matrix $A$ is symmetric. Thus, the corollary follows directly from Proposition \ref{M-equivalent}.
\end{proof}

Let $S_{-1}$ be a set of solutions of the equation $x^2 +y^2 =1$ over $GF(q)$. Then the cardinality of $S_{-1}$ for an odd prime $q$ is obtained in the next proposition.%, which is proved in \cite{park2011classification}.

\begin{proposition}[\cite{park2011classification}]\label{S_-1}  Let $GF(q)$ be a finite field of order $q$ such that $q$ is a power of an odd prime. The cardinality of the set
$$S_{-1}=\{(x,y)\in GF(q)^2\mid x^2+y^2+1=0\}
$$ is given by
$$|S_{-1}|=q-(-1)^{(q-1)/2}=\begin{cases}
q-1, &\text{if $q\equiv 1\pmod{4}$},\\
q+1, &\text{if $q\equiv 3\pmod4$}.
\end{cases}
$$
\end{proposition}

Similarly, we define a set $S_{-I_2}$ of $2 \times 2$ symmetric matrices over $GF(q)$  satisfying the matrix equation $X^2 + I_2 = 0$. We also obtain the cardinality of $S_{-I_2}$ in the following corollary.

\begin{corollary}\label{S_-I_2} Let $S_{-I_2}$ be a set of $2 \times 2$ symmetric matrices over $GF(q)$ where $q$ is a power of odd prime such that
$$S_{-I_2}=\{P\in K \mid P^2=-I_2\}.
$$ Then, the cardinality of $S_{-I_2}$ is given by
$$|S_{-I_2}|=q-(-1)^{( q-1)/2}=\begin{cases}
q-1, &\text{if $q\equiv 1\pmod{4}$},\\
q+1, &\text{if $q\equiv 3\pmod4$}.
\end{cases}
$$

\end{corollary}

\begin{proof}
The condition $P^2=-I_2$ implies that $P^{-1}= - P$. Since we assumed that $P$ is symmetric, it is easy to show that matrix $P$ is in the form $\begin{smatrix}
\alpha & \beta \\ \beta & - \alpha
\end{smatrix}
$, where  $(\alpha, \beta)$ is a solution of the equation $x^2 +y^2 =1$. Thus, the result follows with Proposition \ref{S_-1}.
\end{proof}

\section{Construction method of symmetric self-dual codes}

It is well-known that a self-dual over $GF(q)$ of length $n$ for $q \equiv 1 {\pmod 4}$ exists if and only if $n \equiv 0 {\pmod 2}$, and a self-dual over $GF(q)$ of length $n$ for $q \equiv 3 {\pmod 4}$ exists if and only if $n \equiv 0 {\pmod 4}$ \cite[Theorem 9.1.3]{HP3}. In \cite{choi2020}, we have introduced a construction method for symmetric self-dual codes over $GF(q)$ for $q \equiv 1 \pmod 4$. In this section, we introduce two new construction methods for symmetric self-dual codes over $GF(q)$ for $q \equiv 3 \pmod 4$. These methods generate symmetric self-dual codes of lengths increased by four.

\begin{theorem}[Construction method 1]{\label{SymBuildingup1}}
Let
$G=(I_n \mid A)$
be a generator matrix of symmetric self-dual code $\CC$ of length $2n$ over $GF(q)$ for an odd prime power $q$. Assume that there exists a codeword $(\x_n , \y_n)$ in $\CC$ satisfying $\x_n \cdot \y_n = 0$, $\x_n \cdot \x_n=k(\ne 0)$, and $-1\pm k$ are squares in $GF(q)$. Then, take an element $(\alpha, \beta)$ in $S_{-1}$ and let $B= \left( \begin{matrix}
\alpha \x_n + \beta \y_n \\
\beta \x_n - \alpha \y_n 					
\end{matrix}\right)$,
$E= \frac{1}{k} ( s \x_n^T \x_n  + t \y_n^T \y_n -  \x_n^T \y_n - \y_n^T \x_n)$ where $s^2 = -1+k$ and $t^2=-1-k$, and let $D = -\frac{1}{k^2}B(A+E_1)B^T B B^T$. Then
$$
G_1=	(I_{n+2} \mid A_1)=\left(\begin{array}{c|c|c|c}
I_2 & O& D & B \\
\hline
O& I_n&B^T& A+E\\
\end{array}\right)$$
is a generator matrix of a symmetric self-dual code of length $2n+4$.

\end{theorem}

The proof of Theorem \ref{SymBuildingup1} is given in Appendix for the brevity. We need following two lemmas to introduce the second construction method.

\begin{lemma}{\label{GP1}}
%Let
%$G=(I_n \mid A)$
%be a generator matrix of symmetric self-dual code $\CC$ of length $2n$ %over $GF(p)$ for an odd prime $p$ and let $S_{-I_2}$ be the set defined in proposition \ref{S_-I_2}, and
Let $P=\begin{smatrix}
\alpha & \beta \\ \beta & -\alpha
\end{smatrix}
$ be an element in $S_{-I_2}$ and $A$ be a symmetric matrix satisfying $A^2 = - I_n$. For a vector $\x$ in $GF(q)^n$, if we let the matrix $M = \left( \begin{matrix}
\x\\\beta^{-1} \x (A - \alpha I)
\end{matrix}\right)$, then $$MA=PM.$$

\end{lemma}

\begin{proof}
Let $\y =\beta^{-1} \x (A - \alpha I)$. Then $\beta \y =\x A - \alpha \x$ and this implies that $\x A = \alpha \x + \beta \y.$
On the other hand,

\begin{align*}
	\y(A+\alpha I) &=\beta^{-1} \x (A - \alpha I)(A+\alpha I)\\
	&=\beta^{-1} \x (A^2  - \alpha^2 I) \\
	&=\beta^{-1}\x (-1 - \alpha^2) I \\
	&=\beta \x, {\mbox{ since }} \alpha^2 + \beta^2 = -1
\end{align*}
and this implied that $\y A=\beta \x - \alpha \y $.
Therefore,
$$
M A= \left( \begin{matrix}
\x A\\ \y A
\end{matrix}\right) = \left( \begin{matrix}
\alpha \x + \beta \y \\ \beta \x - \alpha \y
\end{matrix}\right)  =PM.$$
\end{proof}

\begin{lemma}{\label{GP2}}
%Let
%$G=(I_n \mid A)$
%be a generator matrix of symmetric self-dual code $\CC$ of length $2n$ %over $GF(p)$ for an odd prime $p$ and let $S_{-I_2}$ be the set defined in proposition \ref{S_-I_2}, and
Assume that $n \times n$ matrices $H$ and $P$ are symmetric. If $(H+P)(H-P)$ is also symmetric, then $HP=PH$.

\end{lemma}

\begin{proof}
By the assumption, we have
\begin{align*}
	(H-P)(H+P)&=\{(H-P)(H+P)\}^T \\
	&= (H+P)^T (H-P)^T \\
	&=(H+P)(H-P),
\end{align*} and by equating both sides, the result follows.\end{proof}

Now, we give the next theorem, which introduces the second construction method.

\begin{theorem}[Construction 2]{\label{SymBuildingup2}}
Let
$G=(I_n \mid A)$
be a generator matrix of a symmetric self-dual code $\CC$ of length $2n$ over $GF(p)$ for an odd prime $p$ and let $S_{-I_2}$ be the set defined in proposition \ref{S_-I_2}, and let $P=\begin{smatrix}
\alpha & \beta \\ \beta & -\alpha
\end{smatrix}
$ be an element in $S_{-I_2}$. Let $M = \left( \begin{matrix}
\x\\\beta^{-1} \x (A - \alpha I)
\end{matrix}\right)$ for a vector $\x$ in $GF(q)^n$.
Assume that $H$ is a $2\times 2$ symmetric matrix satisfying the equation \begin{equation}\label{HP=PH}(H+P)(H-P)=-M M^T,\end{equation} and $H-P$ is non-singular. Then
$$
G_2=	(I_{n+2} \mid A_2)=\left(\begin{array}{c|c|c|c}
I_2 & O& H & M\\
\hline

O& I_n&M^T& A+M^T (H-P)^{-1} M\\
\end{array}\right)$$
is a generator matrix of a symmetric self-dual code of length $2n+4$.

\end{theorem}

The proof of Theorem \ref{SymBuildingup2} is also given in Appendix. We illustrate these new construction methods in the following examples.

\begin{example}\label{BU1}

Let $\CC_{3}^{8}$ be a symmetric optimal self-dual [8,4,3] code over $GF(3)$ with generator matrix
$$G=\begin{smatrix}
1& 0& 0& 0& 1& 1& 0& 0\\
0& 1& 0& 0& 1& 2& 0& 0\\
0& 0& 1& 0& 0& 0& 2& 1\\
0& 0& 0& 1& 0& 0& 1& 1
\end{smatrix}.$$
To apply construction method in Theorem \ref{SymBuildingup1}, take $(\alpha , \beta) =(1,1)$ and the codeword  $(\x_n| \y_n) = (2,1,1,1,0,1,0,2)$ in $\CC_{3}^{8}$. Then, we compute that
\begin{center}
$B=\begin{smatrix}
2& 2& 1& 0\\
2& 0& 1& 2
\end{smatrix}$,
$D=\begin{smatrix}
2& 1\\
1& 2
\end{smatrix}$, and
$E=\begin{smatrix}
0& 1& 0& 2\\
1& 2& 2& 2\\
0& 2& 0& 1\\
2& 2& 1& 0
\end{smatrix}$.		
\end{center}
Finally, we find an optimal symmetric self-dual [12,6,6] over $GF(3)$ code with generator matrix
$$
G_1= \begin{smatrix}
1& 0& 0& 0& 0& 0& 2& 1& 2& 2& 1& 0\\
0& 1& 0& 0& 0& 0& 1& 2& 2& 0& 1& 2\\
0& 0& 1& 0& 0& 0& 2& 2& 1& 2& 0& 2\\
0& 0& 0& 1& 0& 0& 2& 0& 2& 1& 2& 2\\
0& 0& 0& 0& 1& 0& 1& 1& 0& 2& 2& 2\\
0& 0& 0& 0& 0& 1& 0& 2& 2& 2& 2& 1\\
\end{smatrix}.$$
\end{example}

\begin{example}\label{ex3.6}
Let $\CC_{19}^{8}$ be a symmetric self-dual [8,4,3] code over $GF(19)$ with generator matrix
$$G=\begin{smatrix}
1& 0& 0& 0& 18& 13& 0& 0\\
0& 1& 0& 0& 13& 1& 0& 0\\
0& 0& 1& 0& 0& 0& 1& 6\\
0& 0& 0& 1& 0& 0& 6& 18
\end{smatrix}.$$
To apply construction method in Theorem \ref{SymBuildingup1}, take $(\alpha , \beta) =(18,6)$ and $\x = (1,6,9,6)$ in $GF(19)^4$. Then,
\begin{center}
$M=\begin{smatrix}
1& 6& 9& 6\\
13& 1& 9& 9
\end{smatrix}$ and
$H=\begin{smatrix}
9& 12\\
12& 13
\end{smatrix}$, 	
\end{center}
and finally, we obtain a symmetric [12,6,7] self-dual code over $GF(19)$ of length 12 with generator matrix
$$
G_2=	 \begin{smatrix}
1 & 0  &0 & 0 & 0 & 0 & 9 &12 & 1 & 6 & 9 & 6\\
0 & 1  &0 & 0 & 0 & 0 &12 &13& 13 & 1 & 9 & 9\\
0 & 0  &1 & 0 & 0 & 0 & 1 &13 & 7 &17 &13& 14\\
0 & 0 & 0 & 1 & 0 & 0 & 6 & 1 &17 &14 & 7 & 6\\
0 & 0 & 0 & 0  &1 & 0 & 9 & 9& 13&  7& 12 &11\\
0 & 0 & 0 & 0 & 0 & 1 & 6 & 9& 14 & 6& 11 & 2\\
\end{smatrix}.$$

\end{example}

\section{Computational results}
In this section, we discuss computational results of symmetric self-dual codes over $GF(q)$ for $q=11,19,23$. Using construction methods in Theorem \ref{SymBuildingup1} and \ref{SymBuildingup2}, we obtain many new symmetric self-dual codes of lengths $n \le 40$ which meet the best known bounds on minimum distances of self-dual codes.

We find best symmetric self-dual codes of length $n$ over $GF(q)$ for $q=11,19,23$ and $n\le 40$ except for the case that $q=11$ with $n=16$ or $20$, for the case that $q=19$ with $n=32$, and for the case that $q=23$ with $n=20$ or $24$. Moreover, we also find more than 151 self-dual codes with new parameters: 90 inequivalent self-dual codes of length 32, 36 and 40 over $GF(11)$, 5 inequivalent self-dual codes of length 36 and 40 over $GF(19)$ and 56 inequivalent self-dual codes of length 32, 26 and 40 over $GF(23)$. Among them, we introduce five symmetric self-dual codes with their generator matrices in this section.

At the end of this section, we summarize the known bounds on the highest minimum distances of self-dual codes in Table \ref{summary}.

%\subsection{Symmetric self-dual codes over $GF(7)$. }
%
%
%
%	\begin{proposition}
%		
%	For lengths $n=4,8,12,16, 24, 28$ and $40$, there exist symmetric self-dual codes over $GF(7)$ which have the best-known minimum distance. Moreover, for $n=4,8,12$ and $16$, they are optimal.
%	
%	\end{proposition}
%
%
%\begin{table}[h]
%	\begin{center}
%		\begin{small}
%			\begin{tabular}{cccc|cccc}
%				\hline
%			$[n,k,d]_p$ & $d_{sym.}$ &$d_{s.d.}$&$d_{gen.}$&$[n,k,d]_p$ & $d_{sym.}$ &$d_{s.d.}$&$d_{gen.}$   \\
%				\hline
%			$[4,2,3]_7$ & $3$ &$3$&$3$&$[24,12,9]_7$ & $9$ &$9-11$&$10-11$  \\
%			$[8,4,5]_7$ & $5$ &$5$&$5$&$[28,14,10]_7$ & $10$ &$11-13$&$11-13$  \\
%			$[12,6,6]_7$ & $6$ &$6$&$6$&$[32,16,11]_7$ & $11$ &$13-14$&$13-14$  \\
%			$[16,8,7]_7$ & $7$ &$7$&$7$&$[36,18,12]_7$ & $12$ &$13-16$&$13-16$ \\			
%			$[20,10,8]_7$ & $8$ &$9$&$9$&$[40,20,13]_7$ & $13$ &$13-18$&$14-18$  \\
%		       \hline
%			\end{tabular}
%		\end{small}
%		\caption{Bounds on the minimum distance of symmetric self-dual codes over $GF(7)$ }
%		\label{previous_results_7}
%	\end{center}
%\end{table}

\subsection{Symmetric self-dual codes over $GF(11)$. }

\begin{proposition}

There exist best symmetric self-dual codes over $GF(11)$ of length $n=4,8,12, 24, 28, 32, 36, 40$. In particular, $[4,2,3]_{11}, [8,4,5]_{11}$ and $[12,6,7]_{11}$ symmetric self-dual codes are MDS. Moreover, $[32,16,12]_{11}$, $[36,18,13]_{11}$ and $[40,20,14]_{11}$ codes are new.
\end{proposition}

We give the highest minimum distance $d_{sym}$ of symmetric self-dual codes and the previously best known minimum distance $d_{sd}$ of self-dual codes in Table \ref{previous_results_11}. In this table, new parameters are written in bold. We present three symmetric self-dual codes having new parameters:

\begin{itemize}

\item $[32,16,12]_{11}$ code with a generator matrix $(I_{16} \mid A_{11}^{32})$ where \\
$A_{11}^{32}=\begin{smatrix}
6 & 7 & 7 & 1 & 2 & 8 & 5 & 9 & 9 & 8 & 1 & 6 & 4 & 7 & 10 & 6 \\
7 & 9 & 7 & 8 & 0 & 8 & 4 & 8 & 6 & 10 & 2 & 6 & 9 & 7 & 8 & 10 \\
7 & 7 & 6 & 0 & 8 & 2 & 4 & 9 & 1 & 6 & 8 & 7 & 6 & 9 & 0 & 4 \\
1 & 8 & 0 & 7 & 0 & 7 & 10 & 2 & 1 & 9 & 9 & 3 & 3 & 2 & 8 & 0 \\
2 & 0 & 8 & 0 & 10 & 10 & 8 & 10 & 3 & 0 & 10 & 8 & 0 & 8 & 10 & 0 \\
8 & 8 & 2 & 7 & 10 & 7 & 10 & 2 & 9 & 7 & 7 & 0 & 6 & 1 & 0 & 3 \\
5 & 4 & 4 & 10 & 8 & 10 & 10 & 7 & 8 & 5 & 2 & 5 & 4 & 8 & 3 & 9 \\
9 & 8 & 9 & 2 & 10 & 2 & 7 & 0 & 3 & 2 & 8 & 10 & 7 & 8 & 4 & 6 \\
9 & 6 & 1 & 1 & 3 & 9 & 8 & 3 & 0 & 1 & 5 & 10 & 7 & 7 & 8 & 10 \\
8 & 10 & 6 & 9 & 0 & 7 & 5 & 2 & 1 & 9 & 9 & 1 & 1 & 9 & 4 & 4 \\
1 & 2 & 8 & 9 & 10 & 7 & 2 & 8 & 5 & 9 & 3 & 7 & 7 & 2 & 9 & 4 \\
6 & 6 & 7 & 3 & 8 & 0 & 5 & 10 & 10 & 1 & 7 & 7 & 2 & 6 & 8 & 2 \\
4 & 9 & 6 & 3 & 0 & 6 & 4 & 7 & 7 & 1 & 7 & 2 & 5 & 7 & 7 & 6 \\
7 & 7 & 9 & 2 & 8 & 1 & 8 & 8 & 7 & 9 & 2 & 6 & 7 & 4 & 1 & 5 \\
10 & 8 & 0 & 8 & 10 & 0 & 3 & 4 & 8 & 4 & 9 & 8 & 7 & 1 & 7 & 2 \\
6 & 10 & 4 & 0 & 0 & 3 & 9 & 6 & 10 & 4 & 4 & 2 & 6 & 5 & 2 & 9
\end{smatrix}$

\item $[36,18,13]_{11}$ code with a generator matrix $(I_{18} \mid A_{11}^{36})$ where \\
$A_{11}^{36}=\begin{smatrix}
5 & 10 & 6 & 7 & 5 & 4 & 7 & 1 & 9 & 4 & 4 & 7 & 7 & 8 & 1 & 8 & 8 & 8 \\
10 & 10 & 8 & 7 & 6 & 4 & 5 & 8 & 9 & 7 & 10 & 1 & 3 & 0 & 5 & 8 & 9 & 9 \\
6 & 8 & 10 & 6 & 4 & 2 & 4 & 6 & 10 & 0 & 1 & 5 & 6 & 9 & 1 & 9 & 5 & 1 \\
7 & 7 & 6 & 3 & 4 & 3 & 4 & 2 & 1 & 10 & 4 & 1 & 5 & 3 & 7 & 8 & 4 & 6 \\
5 & 6 & 4 & 4 & 6 & 3 & 7 & 4 & 8 & 9 & 9 & 9 & 9 & 10 & 0 & 4 & 5 & 9 \\
4 & 4 & 2 & 3 & 3 & 1 & 7 & 2 & 2 & 0 & 3 & 7 & 6 & 6 & 5 & 1 & 1 & 4 \\
7 & 5 & 4 & 4 & 7 & 7 & 6 & 4 & 10 & 7 & 1 & 2 & 9 & 1 & 4 & 0 & 6 & 7 \\
1 & 8 & 6 & 2 & 4 & 2 & 4 & 3 & 7 & 8 & 4 & 1 & 1 & 1 & 2 & 2 & 1 & 4 \\
9 & 9 & 10 & 1 & 8 & 2 & 10 & 7 & 3 & 8 & 7 & 6 & 9 & 9 & 3 & 3 & 1 & 7 \\
4 & 7 & 0 & 10 & 9 & 0 & 7 & 8 & 8 & 0 & 5 & 0 & 0 & 8 & 0 & 4 & 8 & 10 \\
4 & 10 & 1 & 4 & 9 & 3 & 1 & 4 & 7 & 5 & 2 & 10 & 3 & 3 & 2 & 0 & 1 & 3 \\
7 & 1 & 5 & 1 & 9 & 7 & 2 & 1 & 6 & 0 & 10 & 0 & 0 & 8 & 6 & 5 & 0 & 0 \\
7 & 3 & 6 & 5 & 9 & 6 & 9 & 1 & 9 & 0 & 3 & 0 & 9 & 7 & 6 & 7 & 5 & 0 \\
8 & 0 & 9 & 3 & 10 & 6 & 1 & 1 & 9 & 8 & 3 & 8 & 7 & 8 & 8 & 1 & 5 & 10 \\
1 & 5 & 1 & 7 & 0 & 5 & 4 & 2 & 3 & 0 & 2 & 6 & 6 & 8 & 10 & 0 & 8 & 7 \\
8 & 8 & 9 & 8 & 4 & 1 & 0 & 2 & 3 & 4 & 0 & 5 & 7 & 1 & 0 & 2 & 9 & 2 \\
8 & 9 & 5 & 4 & 5 & 1 & 6 & 1 & 1 & 8 & 1 & 0 & 5 & 5 & 8 & 9 & 4 & 10 \\
8 & 9 & 1 & 6 & 9 & 4 & 7 & 4 & 7 & 10 & 3 & 0 & 0 & 10 & 7 & 2 & 10 & 6
\end{smatrix}$

\item $[40,20,14]_{11}$ code with a generator matrix $(I_{20} \mid A_{11}^{40})$ where \\
$A_{11}^{40}=\begin{smatrix}
5 & 4 & 6 & 1 & 7 & 10 & 5 & 5 & 8 & 8 & 10 & 4 & 9 & 5 & 9 & 5 & 8 & 3 & 9 & 6 \\
4 & 2 & 6 & 3 & 1 & 2 & 1 & 2 & 5 & 8 & 2 & 4 & 5 & 9 & 1 & 7 & 5 & 7 & 3 & 4 \\
6 & 6 & 2 & 9 & 4 & 4 & 9 & 5 & 1 & 8 & 6 & 2 & 6 & 9 & 10 & 5 & 6 & 0 & 5 & 0 \\
1 & 3 & 9 & 5 & 3 & 10 & 2 & 4 & 10 & 3 & 1 & 10 & 9 & 7 & 8 & 9 & 10 & 7 & 0 & 0 \\
7 & 1 & 4 & 3 & 7 & 3 & 8 & 0 & 1 & 7 & 7 & 6 & 0 & 5 & 7 & 10 & 9 & 6 & 5 & 0 \\
10 & 2 & 4 & 10 & 3 & 6 & 0 & 9 & 3 & 9 & 4 & 8 & 9 & 0 & 3 & 4 & 6 & 8 & 0 & 5 \\
5 & 1 & 9 & 2 & 8 & 0 & 4 & 6 & 1 & 2 & 4 & 8 & 6 & 3 & 8 & 5 & 5 & 4 & 3 & 3 \\
5 & 2 & 5 & 4 & 0 & 9 & 6 & 3 & 4 & 7 & 8 & 8 & 3 & 2 & 10 & 3 & 2 & 3 & 3 & 7 \\
8 & 5 & 1 & 10 & 1 & 3 & 1 & 4 & 0 & 7 & 9 & 3 & 2 & 9 & 9 & 2 & 9 & 9 & 1 & 6 \\
8 & 8 & 8 & 3 & 7 & 9 & 2 & 7 & 7 & 5 & 3 & 9 & 3 & 0 & 5 & 8 & 5 & 6 & 8 & 8 \\
10 & 2 & 6 & 1 & 7 & 4 & 4 & 8 & 9 & 3 & 6 & 6 & 1 & 5 & 7 & 10 & 2 & 3 & 8 & 6 \\
4 & 4 & 2 & 10 & 6 & 8 & 8 & 8 & 3 & 9 & 6 & 6 & 10 & 5 & 2 & 6 & 7 & 6 & 6 & 1 \\
9 & 5 & 6 & 9 & 0 & 9 & 6 & 3 & 2 & 3 & 1 & 10 & 4 & 9 & 4 & 7 & 3 & 2 & 8 & 10 \\
5 & 9 & 9 & 7 & 5 & 0 & 3 & 2 & 9 & 0 & 5 & 5 & 9 & 2 & 6 & 8 & 2 & 10 & 8 & 0 \\
9 & 1 & 10 & 8 & 7 & 3 & 8 & 10 & 9 & 5 & 7 & 2 & 4 & 6 & 1 & 3 & 6 & 1 & 7 & 4 \\
5 & 7 & 5 & 9 & 10 & 4 & 5 & 3 & 2 & 8 & 10 & 6 & 7 & 8 & 3 & 4 & 9 & 2 & 5 & 3 \\
8 & 5 & 6 & 10 & 9 & 6 & 5 & 2 & 9 & 5 & 2 & 7 & 3 & 2 & 6 & 9 & 4 & 9 & 3 & 6 \\
3 & 7 & 0 & 7 & 6 & 8 & 4 & 3 & 9 & 6 & 3 & 6 & 2 & 10 & 1 & 2 & 9 & 6 & 7 & 1 \\
9 & 3 & 5 & 0 & 5 & 0 & 3 & 3 & 1 & 8 & 8 & 6 & 8 & 8 & 7 & 5 & 3 & 7 & 3 & 10 \\
6 & 4 & 0 & 0 & 0 & 5 & 3 & 7 & 6 & 8 & 6 & 1 & 10 & 0 & 4 & 3 & 6 & 1 & 10 & 2
\end{smatrix}$
\end{itemize}

\begin{table}[h]
\begin{center}
\begin{small}
	\begin{tabular}{ccc|ccc}
		\hline
		$[n,k,d]_p$ & $d_{sym.}$ &$d_{sd.}$&$[n,k,d]_p$ & $d_{sym.}$ &$d_{sd.}$  \\
		\hline
		$[4,2,3]_{11}$ & $3$ &$3$&$[24,12,9]_{11}$ & $9$ &$9-12$  \\
		$[8,4,5]_{11}$ & $5$ &$5$&$[28,14,10]_{11}$ & $10$ &$10-14$  \\
		$[12,6,7]_{11}$ & $7$ &$7$&$\mathbf{[32,16,12]_{11}}$ & $\mathbf{12}$&$?-16$  \\
		$[16,7,8]_{11}$ & $8$ &$9$&$\mathbf{[36,18,13]_{11}}$ & $\mathbf{13}$&$12-18$ \\			
		$[20,10,8]_{11}$ & $8$ &$10$&$\mathbf{[40,20,14]_{11}}$ & $\mathbf{14}$ &$13-20$  \\
		\hline
	\end{tabular}
\end{small}
\caption{Best known minimum distances of symmetric self-dual codes over $GF(11)$ }
\label{previous_results_11}
\end{center}
\end{table}

\subsection{Symmetric self-dual codes over $GF(19)$. }
\begin{proposition}
There exist best symmetric self-dual codes over $GF(19)$ of length \linebreak $n=4,8,12,16,20,24,28,36,40$. Among them, $[4,2,3]_{19}$, $[8,4,5]_{19}$, $[12,6,7]_{19}$, and $[20,10,11]_{19}$ codes are MDS. Moreover, $[36,18,16]_{19}$ and $[40, 20,15]_{19}$ codes are new.
\end{proposition}

We give the highest minimum distance $d_{sym}$ of symmetric self-dual codes and the previously best known minimum distance $d_{sd}$ of self-dual codes in Table \ref{previous_results_19}. In this table, new parameters are written in bold. We present two symmetric self-dual codes having new parameters:
\begin{itemize}

\item $[36,18,14]_{19}$ code with a generator matrix $(I_{18} \mid A_{19}^{36})$ where \\
$A_{19}^{36}=\begin{smatrix}
16 & 14 & 4 & 15 & 10 & 15 & 17 & 7 & 4 & 16 & 16 & 14 & 3 & 7 & 5 & 2 & 5 & 12 \\
14 & 13 & 18 & 11 & 15 & 17 & 11 & 5 & 5 & 11 & 16 & 16 & 12 & 4 & 17 & 0 & 16 & 4 \\
4 & 18 & 18 & 12 & 18 & 12 & 2 & 6 & 12 & 18 & 14 & 1 & 10 & 16 & 10 & 6 & 13 & 6 \\
15 & 11 & 12 & 7 & 1 & 8 & 1 & 3 & 1 & 12 & 11 & 5 & 5 & 7 & 7 & 2 & 10 & 8 \\
10 & 15 & 18 & 1 & 7 & 9 & 14 & 14 & 7 & 12 & 13 & 16 & 16 & 2 & 16 & 9 & 16 & 4 \\
15 & 17 & 12 & 8 & 9 & 2 & 7 & 15 & 5 & 12 & 2 & 9 & 2 & 10 & 14 & 18 & 12 & 9 \\
17 & 11 & 2 & 1 & 14 & 7 & 11 & 13 & 16 & 1 & 16 & 17 & 4 & 11 & 4 & 11 & 9 & 18 \\
7 & 5 & 6 & 3 & 14 & 15 & 13 & 9 & 0 & 16 & 3 & 3 & 8 & 7 & 10 & 14 & 4 & 7 \\
4 & 5 & 12 & 1 & 7 & 5 & 16 & 0 & 12 & 17 & 1 & 7 & 4 & 0 & 9 & 0 & 17 & 18 \\
16 & 11 & 18 & 12 & 12 & 12 & 1 & 16 & 17 & 3 & 1 & 17 & 12 & 12 & 16 & 12 & 9 & 11 \\
16 & 16 & 14 & 11 & 13 & 2 & 16 & 3 & 1 & 1 & 15 & 17 & 4 & 12 & 10 & 0 & 7 & 4 \\
14 & 16 & 1 & 5 & 16 & 9 & 17 & 3 & 7 & 17 & 17 & 8 & 7 & 6 & 18 & 1 & 11 & 18 \\
3 & 12 & 10 & 5 & 16 & 2 & 4 & 8 & 4 & 12 & 4 & 7 & 14 & 0 & 9 & 9 & 6 & 15 \\
7 & 4 & 16 & 7 & 2 & 10 & 11 & 7 & 0 & 12 & 12 & 6 & 0 & 18 & 5 & 13 & 13 & 4 \\
5 & 17 & 10 & 7 & 16 & 14 & 4 & 10 & 9 & 16 & 10 & 18 & 9 & 5 & 13 & 7 & 0 & 7 \\
2 & 0 & 6 & 2 & 9 & 18 & 11 & 14 & 0 & 12 & 0 & 1 & 9 & 13 & 7 & 12 & 17 & 3 \\
5 & 16 & 13 & 10 & 16 & 12 & 9 & 4 & 17 & 9 & 7 & 11 & 6 & 13 & 0 & 17 & 2 & 17 \\
12 & 4 & 6 & 8 & 4 & 9 & 18 & 7 & 18 & 11 & 4 & 18 & 15 & 4 & 7 & 3 & 17 & 10
\end{smatrix}$

\item $[40,20,15]_{19}$ code with a generator matrix $(I_{20} \mid A_{19}^{40})$ where \\
$A_{19}^{40}=\begin{smatrix}
8 & 7 & 12 & 5 & 0 & 13 & 15 & 11 & 16 & 6 & 17 & 14 & 6 & 6 & 6 & 4 & 18 & 14 & 13 & 14 \\
7 & 16 & 13 & 4 & 13 & 3 & 1 & 13 & 4 & 11 & 5 & 12 & 6 & 4 & 13 & 16 & 11 & 6 & 6 & 16 \\
12 & 13 & 5 & 6 & 13 & 11 & 12 & 12 & 16 & 9 & 3 & 0 & 16 & 17 & 2 & 8 & 6 & 14 & 6 & 9 \\
5 & 4 & 6 & 1 & 7 & 15 & 0 & 4 & 16 & 14 & 1 & 8 & 0 & 9 & 12 & 9 & 10 & 5 & 16 & 2 \\
0 & 13 & 13 & 7 & 17 & 17 & 18 & 17 & 16 & 16 & 8 & 0 & 1 & 13 & 14 & 3 & 11 & 6 & 9 & 14 \\
13 & 3 & 11 & 15 & 17 & 9 & 0 & 0 & 11 & 2 & 11 & 1 & 11 & 13 & 15 & 16 & 16 & 15 & 0 & 0 \\
15 & 1 & 12 & 0 & 18 & 0 & 0 & 14 & 4 & 15 & 13 & 8 & 14 & 17 & 17 & 9 & 0 & 5 & 15 & 0 \\
11 & 13 & 12 & 4 & 17 & 0 & 14 & 6 & 12 & 5 & 18 & 18 & 12 & 8 & 3 & 13 & 15 & 10 & 11 & 18 \\
16 & 4 & 16 & 16 & 16 & 11 & 4 & 12 & 7 & 2 & 10 & 3 & 4 & 16 & 1 & 13 & 16 & 8 & 12 & 17 \\
6 & 11 & 9 & 14 & 16 & 2 & 15 & 5 & 2 & 12 & 12 & 12 & 8 & 5 & 14 & 3 & 5 & 1 & 17 & 9 \\
17 & 5 & 3 & 1 & 8 & 11 & 13 & 18 & 10 & 12 & 18 & 0 & 16 & 8 & 11 & 18 & 5 & 17 & 10 & 10 \\
14 & 12 & 0 & 8 & 0 & 1 & 8 & 18 & 3 & 12 & 0 & 16 & 1 & 11 & 14 & 10 & 14 & 7 & 7 & 17 \\
6 & 6 & 16 & 0 & 1 & 11 & 14 & 12 & 4 & 8 & 16 & 1 & 6 & 10 & 7 & 13 & 9 & 6 & 4 & 0 \\
6 & 4 & 17 & 9 & 13 & 13 & 17 & 8 & 16 & 5 & 8 & 11 & 10 & 5 & 5 & 9 & 11 & 7 & 13 & 4 \\
6 & 13 & 2 & 12 & 14 & 15 & 17 & 3 & 1 & 14 & 11 & 14 & 7 & 5 & 6 & 13 & 10 & 15 & 14 & 8 \\
4 & 16 & 8 & 9 & 3 & 16 & 9 & 13 & 13 & 3 & 18 & 10 & 13 & 9 & 13 & 13 & 15 & 17 & 0 & 17 \\
18 & 11 & 6 & 10 & 11 & 16 & 0 & 15 & 16 & 5 & 5 & 14 & 9 & 11 & 10 & 15 & 8 & 6 & 18 & 17 \\
14 & 6 & 14 & 5 & 6 & 15 & 5 & 10 & 8 & 1 & 17 & 7 & 6 & 7 & 15 & 17 & 6 & 7 & 11 & 17 \\
13 & 6 & 6 & 16 & 9 & 0 & 15 & 11 & 12 & 17 & 10 & 7 & 4 & 13 & 14 & 0 & 18 & 11 & 13 & 5 \\
14 & 16 & 9 & 2 & 14 & 0 & 0 & 18 & 17 & 9 & 10 & 17 & 0 & 4 & 8 & 17 & 17 & 17 & 5 & 17
\end{smatrix}$
\end{itemize}

\begin{table}[h]
\begin{center}
\begin{small}
	\begin{tabular}{ccc|ccc}
		\hline
		$[n,k,d]_p$ & $d_{sym}$ &$d_{sd}$&$[n,k,d]_p$ & $d_{sym}$ &$d_{sd}$  \\
		\hline
		$[4,2,3]_{19}$ & $3$ &$3$&$[24,12,10]_{19}$ & $10$ &$10-12$  \\
		$[8,4,5]_{19}$ & $5$ &$5$&$[28,14,11]_{19}$ & $11$ &$11-14$  \\
		$[12,6,7]_{19}$ & $7$ &$7$&$[32,16,12]_{19}$ & $12$&$14-16$  \\
		$[16,8,8]_{19}$ & $8$ &$8-9$&$\mathbf{[36,18,14]_{19}}$ & $\mathbf{14}$&$?-18$ \\			
		$[20,10,11]_{19}$ & $11$ &$11$&$\mathbf{[40,20,15]_{19}}$ & $\mathbf{15}$ &$?-20$  \\
		\hline
	\end{tabular}
\end{small}
\caption{Best known minimum distances of self-dual codes over $GF(19)$ }
\label{previous_results_19}
\end{center}
\end{table}

\subsection{Symmetric self-dual codes over $GF(23)$. }

\begin{proposition}
There exist the best symmetric self-dual codes over $GF(23)$ of length $n=4,8,12, 28,32,36,40$. Among them, $[4,2,3]_{23}$, $[8,4,5]_{23}$ and $[12,6,7]_{23}$ codes are MDS. Moreover, $[32,16,12]_{23}$, $[36,18,16]_{19}$ and $[40, 20,15]_{19}$ codes are new.
\end{proposition}

We give the highest minimum distance $d_{sym}$ of symmetric self-dual codes and the previously best known minimum distance $d_{sd}$ of self-dual codes in Table \ref{previous_results_23}. In this table, new parameters are written in bold. We present three symmetric self-dual codes having new parameters:

\begin{itemize}

\item $[32,16,12]_{23}$ code with a generator matrix $(I_{16} \mid A_{23}^{32})$ where \\
$A_{23}^{32}=\begin{smatrix}
20 & 4 & 11 & 18 & 21 & 7 & 19 & 7 & 15 & 6 & 18 & 18 & 2 & 10 & 5 & 12 \\
4 & 1 & 19 & 12 & 11 & 19 & 20 & 8 & 10 & 3 & 11 & 0 & 3 & 6 & 18 & 18 \\
11 & 19 & 12 & 20 & 9 & 2 & 0 & 22 & 21 & 21 & 9 & 6 & 21 & 16 & 13 & 9 \\
18 & 12 & 20 & 2 & 13 & 7 & 4 & 7 & 22 & 18 & 5 & 15 & 0 & 5 & 11 & 20 \\
21 & 11 & 9 & 13 & 9 & 20 & 8 & 19 & 11 & 12 & 11 & 21 & 19 & 14 & 19 & 20 \\
7 & 19 & 2 & 7 & 20 & 3 & 12 & 19 & 5 & 2 & 22 & 1 & 21 & 21 & 22 & 13 \\
19 & 20 & 0 & 4 & 8 & 12 & 12 & 4 & 19 & 7 & 17 & 11 & 8 & 4 & 1 & 0 \\
7 & 8 & 22 & 7 & 19 & 19 & 4 & 1 & 0 & 0 & 16 & 16 & 8 & 15 & 9 & 3 \\
15 & 10 & 21 & 22 & 11 & 5 & 19 & 0 & 22 & 1 & 2 & 12 & 13 & 12 & 10 & 5 \\
6 & 3 & 21 & 18 & 12 & 2 & 7 & 0 & 1 & 14 & 5 & 5 & 22 & 18 & 11 & 1 \\
18 & 11 & 9 & 5 & 11 & 22 & 17 & 16 & 2 & 5 & 10 & 20 & 18 & 12 & 6 & 18 \\
18 & 0 & 6 & 15 & 21 & 1 & 11 & 16 & 12 & 5 & 20 & 13 & 16 & 17 & 5 & 22 \\
2 & 3 & 21 & 0 & 19 & 21 & 8 & 8 & 13 & 22 & 18 & 16 & 3 & 5 & 7 & 6 \\
10 & 6 & 16 & 5 & 14 & 21 & 4 & 15 & 12 & 18 & 12 & 17 & 5 & 7 & 18 & 2 \\
5 & 18 & 13 & 11 & 19 & 22 & 1 & 9 & 10 & 11 & 6 & 5 & 7 & 18 & 3 & 0 \\
12 & 18 & 9 & 20 & 20 & 13 & 0 & 3 & 5 & 1 & 18 & 22 & 6 & 2 & 0 & 6
\end{smatrix}$

\item $[36,18,14]_{23}$ code with a generator matrix $(I_{18} \mid A_{23}^{36})$ where \\
$A_{23}^{36}=\begin{smatrix}
14 & 8 & 18 & 22 & 3 & 17 & 6 & 3 & 2 & 7 & 14 & 2 & 22 & 12 & 6 & 8 & 12 & 19 \\
8 & 18 & 5 & 14 & 21 & 10 & 12 & 16 & 15 & 0 & 18 & 16 & 0 & 20 & 3 & 1 & 2 & 6 \\
18 & 5 & 3 & 13 & 8 & 0 & 20 & 1 & 13 & 12 & 22 & 19 & 14 & 9 & 14 & 15 & 22 & 18 \\
22 & 14 & 13 & 9 & 7 & 1 & 1 & 17 & 10 & 9 & 22 & 14 & 1 & 11 & 11 & 15 & 19 & 12 \\
3 & 21 & 8 & 7 & 20 & 9 & 0 & 2 & 10 & 19 & 7 & 5 & 16 & 0 & 16 & 3 & 15 & 19 \\
17 & 10 & 0 & 1 & 9 & 22 & 1 & 13 & 4 & 13 & 5 & 10 & 14 & 0 & 14 & 17 & 2 & 8 \\
6 & 12 & 20 & 1 & 0 & 1 & 3 & 13 & 20 & 13 & 1 & 19 & 2 & 6 & 12 & 12 & 2 & 0 \\
3 & 16 & 1 & 17 & 2 & 13 & 13 & 21 & 13 & 16 & 17 & 4 & 6 & 6 & 5 & 12 & 9 & 1 \\
2 & 15 & 13 & 10 & 10 & 4 & 20 & 13 & 22 & 2 & 6 & 5 & 14 & 14 & 13 & 16 & 13 & 8 \\
7 & 0 & 12 & 9 & 19 & 13 & 13 & 16 & 2 & 17 & 9 & 21 & 17 & 4 & 2 & 7 & 20 & 7 \\
14 & 18 & 22 & 22 & 7 & 5 & 1 & 17 & 6 & 9 & 20 & 15 & 4 & 22 & 13 & 0 & 17 & 14 \\
2 & 16 & 19 & 14 & 5 & 10 & 19 & 4 & 5 & 21 & 15 & 13 & 2 & 15 & 20 & 20 & 3 & 18 \\
22 & 0 & 14 & 1 & 16 & 14 & 2 & 6 & 14 & 17 & 4 & 2 & 0 & 2 & 9 & 0 & 11 & 1 \\
12 & 20 & 9 & 11 & 0 & 0 & 6 & 6 & 14 & 4 & 22 & 15 & 2 & 1 & 7 & 7 & 5 & 15 \\
6 & 3 & 14 & 11 & 16 & 14 & 12 & 5 & 13 & 2 & 13 & 20 & 9 & 7 & 13 & 1 & 9 & 20 \\
8 & 1 & 15 & 15 & 3 & 17 & 12 & 12 & 16 & 7 & 0 & 20 & 0 & 7 & 1 & 1 & 4 & 9 \\
12 & 2 & 22 & 19 & 15 & 2 & 2 & 9 & 13 & 20 & 17 & 3 & 11 & 5 & 9 & 4 & 18 & 15 \\
19 & 6 & 18 & 12 & 19 & 8 & 0 & 1 & 8 & 7 & 14 & 18 & 1 & 15 & 20 & 9 & 15 & 15
\end{smatrix}$

\item $[40,20,15]_{23}$ code with a generator matrix $(I_{20} \mid A_{23}^{40})$ where \\
$A_{23}^{40}=\begin{smatrix}
3 & 3 & 17 & 18 & 20 & 7 & 20 & 7 & 8 & 12 & 14 & 0 & 8 & 22 & 18 & 0 & 0 & 8 & 9 & 19 \\
3 & 8 & 11 & 22 & 1 & 4 & 11 & 4 & 4 & 9 & 8 & 20 & 7 & 21 & 19 & 16 & 13 & 9 & 22 & 10 \\
17 & 11 & 8 & 10 & 0 & 11 & 3 & 6 & 20 & 11 & 3 & 20 & 15 & 1 & 14 & 4 & 11 & 8 & 19 & 6 \\
18 & 22 & 10 & 1 & 0 & 6 & 14 & 2 & 1 & 15 & 22 & 19 & 11 & 1 & 2 & 3 & 5 & 12 & 12 & 16 \\
20 & 1 & 0 & 0 & 12 & 5 & 0 & 15 & 0 & 5 & 20 & 1 & 13 & 16 & 21 & 14 & 8 & 21 & 10 & 3 \\
7 & 4 & 11 & 6 & 5 & 15 & 2 & 7 & 9 & 0 & 22 & 15 & 1 & 16 & 22 & 2 & 8 & 13 & 16 & 7 \\
20 & 11 & 3 & 14 & 0 & 2 & 17 & 4 & 17 & 6 & 0 & 6 & 14 & 5 & 19 & 8 & 11 & 11 & 17 & 5 \\
7 & 4 & 6 & 2 & 15 & 7 & 4 & 5 & 9 & 19 & 16 & 19 & 7 & 12 & 12 & 14 & 11 & 15 & 22 & 3 \\
8 & 4 & 20 & 1 & 0 & 9 & 17 & 9 & 12 & 10 & 13 & 18 & 5 & 6 & 19 & 10 & 21 & 5 & 5 & 10 \\
12 & 9 & 11 & 15 & 5 & 0 & 6 & 19 & 10 & 5 & 2 & 16 & 14 & 13 & 5 & 6 & 5 & 9 & 20 & 14 \\
14 & 8 & 3 & 22 & 20 & 22 & 0 & 16 & 13 & 2 & 0 & 14 & 10 & 9 & 10 & 19 & 7 & 12 & 7 & 21 \\
0 & 20 & 20 & 19 & 1 & 15 & 6 & 19 & 18 & 16 & 14 & 13 & 18 & 4 & 9 & 10 & 4 & 11 & 2 & 9 \\
8 & 7 & 15 & 11 & 13 & 1 & 14 & 7 & 5 & 14 & 10 & 18 & 15 & 19 & 6 & 20 & 20 & 12 & 15 & 22 \\
22 & 21 & 1 & 1 & 16 & 16 & 5 & 12 & 6 & 13 & 9 & 4 & 19 & 9 & 16 & 6 & 17 & 20 & 8 & 12 \\
18 & 19 & 14 & 2 & 21 & 22 & 19 & 12 & 19 & 5 & 10 & 9 & 6 & 16 & 17 & 14 & 13 & 15 & 3 & 13 \\
0 & 16 & 4 & 3 & 14 & 2 & 8 & 14 & 10 & 6 & 19 & 10 & 20 & 6 & 14 & 15 & 4 & 1 & 9 & 12 \\
0 & 13 & 11 & 5 & 8 & 8 & 11 & 11 & 21 & 5 & 7 & 4 & 20 & 17 & 13 & 4 & 12 & 22 & 18 & 19 \\
8 & 9 & 8 & 12 & 21 & 13 & 11 & 15 & 5 & 9 & 12 & 11 & 12 & 20 & 15 & 1 & 22 & 18 & 14 & 8 \\
9 & 22 & 19 & 12 & 10 & 16 & 17 & 22 & 5 & 20 & 7 & 2 & 15 & 8 & 3 & 9 & 18 & 14 & 20 & 5 \\
19 & 10 & 6 & 16 & 3 & 7 & 5 & 3 & 10 & 14 & 21 & 9 & 22 & 12 & 13 & 12 & 19 & 8 & 5 & 2
\end{smatrix}$
\end{itemize}

\begin{table}[h]
\begin{center}
\begin{small}
	\begin{tabular}{ccc|ccc}
		\hline
		$[n,k,d]_p$ & $d_{sym.}$ &$d_{s.d.}$&$[n,k,d]_p$ & $d_{sym.}$ &$d_{s.d.}$  \\
		\hline
		$[4,2,3]_{23}$ & $3$ &$3$&$[24,12,10]_{19}$ & $10$ &$13$  \\
		$[8,4,5]_{23}$ & $5$ &$5$&$[28,14,11]_{19}$ & $11$ &$11-14$  \\
		$[12,6,7]_{23}$ & $7$ &$7$&$\mathbf{[32,16,12]_{19}}$ & $\mathbf{12}$&$?-16$  \\
		$[16,8,8]_{23}$ & $8$ &$9$&$\mathbf{[36,18,14]_{19}}$ & $\mathbf{14}$&$13-18$ \\			
		$[20,10,9]_{23}$ & $9$ &$10-11$&$\mathbf{[40,20,15]_{19}}$ & $\mathbf{15}$ &$14-20$  \\
		\hline
	\end{tabular}
\end{small}
\caption{Best known minimum distances of symmetric self-dual codes over $GF(23)$ }
\label{previous_results_23}
\end{center}
\end{table}

\begin{table}[h!]
\begin{center}
\begin{small}
	\begin{tabular}{|c|r|r|r|r|r|r|r|}%{|>{\centering\arraybackslash}p{0.5cm}||>{\centering\arraybackslash}p{1.2cm}|>{\centering\arraybackslash}p{1.2cm}|>{\centering\arraybackslash}p{1.2cm}|>{\centering\arraybackslash}p{1.2cm}|>{\centering\arraybackslash}p{1.2cm}|>{\centering\arraybackslash}p{1.2cm}|>{\centering\arraybackslash}p{1.2cm}|>{\centering\arraybackslash}p{1.2cm}|>{\centering\arraybackslash}p{1.2cm}|}
		\hline
		$n\backslash q$ & \multicolumn{1}{|c}{5}    & \multicolumn{1}{|c}{7}      & \multicolumn{1}{|c}{11}             & \multicolumn{1}{|c}{13}       & \multicolumn{1}{|c}{17}               & \multicolumn{1}{|c}{19}       &  \multicolumn{1}{|c|}{23}      \\\hline
		\hline
		2  &$2^*$       &-          &-          &$2^*$       &$2^*$     & -         & -     \\\hline
		4  &$2^o$       &$3^*$		&$3^*$		&$3^*$       &$3^*$     &$3^*$      & $3^*$       \\\hline
		6  &$4^*$       &-          &-          &$4^*$       &$4^*$     & -         & -     \\\hline
		8  &$4^o$       &$5^*$		&$5^*$		&$5^*$       &$5^*$     &$5^*$      & $5^*$       \\\hline
		10 &$4^o$       &-          &-          &$6^*$       &$6^*$     & -         & -     \\\hline
		12 &$6^o$       &$6^o$		&$7^*$		&$6^o$       &$7^*$     &$7^*$      & $7^*$       \\\hline
		14 &$6^o$       &-          &-          &$8^*$       &$7-8$     & -         & -     \\\hline
		16 &$7^o$       & $7^o$	    &$8^o$		&$8^o$      &$8-9$  	& $8-9$     & $9^*$       \\\hline
		18 &$7^o$       &-          & -         &${8-9}$     &$10^*$ 	& -         & -     \\\hline
		20 &$8^o$       & $9^o$	&$10^o$		&$10^o$      &$10^o$ 	&$11^*$     & $10-11$       \\\hline	
		22 &$8^o$       &-          & -         &$10-11$     &$10-11$	& -         & -     \\\hline
		24 &$9-10$      & $9-11$	& $9-12$	&$10-12$   	 &$10-12$	& $10-12$   & $13^*$\\\hline
		26 &$9-10$      &-          & -         & ${10-13}$  &${10-13}$ & -         & -     \\\hline
		28 &$10-11$     & $11-13$	& $10-14$ 	& ${11-14}$  &${11-14}$ & $11-14$   & $11-14$         \\\hline
		30 &$10-12$     &-          & -         & ${11-15}$  &${12-15}$ & -         & -     \\\hline
		32 &$11-13$     & $13-14$	& $12-16$   & ${12-16}$  &${12-16}$ & ${14-16}$ & $12-16$             \\\hline
		34 &$11-14$     &-          &-          & ${12-17}$  &${13-17}$ &-          & -     \\\hline
		36 &$12-15$     & $13-16$	& $13-18$   & ${13-18}$  &${13-18}$ & $14-18$   & $14-18$       \\\hline
		38 &$12-16$     &-          &-          & ${13-19}$  &${14-19}$ &-          & -     \\\hline
		40 &$13-17$     & $14-18$	& $14-20$   & ${14-20}$  &${14-20}$ & $15-20$   & $15-20$       \\\hline
	\end{tabular}
\end{small}
\caption{Bounds on the highest minimum distances of self-dual codes over $GF(p)$ for primes $5 \le p \le 23$ up to lengths 40. In this table, ${}^o$ denotes optimal code and ${}^*$ denotes MDS code.}
\label{summary}
\end{center}
\end{table}

\section{Conclusions}
%
%In this paper, we have constructed symmetric self-dual codes with best-known parameters by using our new construction method. In addition, we improved the bounds of the highest minimum weights of self-dual codes over finite fields.
In this article, we introduced new construction methods of symmetric self-dual codes over finite fields. Then we have constructed many new symmetric self-dual codes, including 153 self-dual codes with new parameters, up to equivalence.
This paper contributes in two ways. One is to provide new \linebreak construction methods of symmetric self-dual codes over $GF(q)$ for the \linebreak challenging case of $q \equiv 3 \pmod 4$. The other is to improve bounds on the highest minimum distance of self-dual codes, which have not been significantly updated for almost two decades because of computational complexity.
We believe that our methods can produce more results for self-dual codes over larger finite fields and/or of longer lengths.

\section*{Acknowledgment}
The author sincerely thanks Dr. Markus Grassl for his helpful comments which was crucial for the implementation.
\appendix

\section{Proof of Theorem \ref{SymBuildingup1}}
\begin{proof}
%	The `particular' part is trivial. We have only to show that $A' (A')^T$ is equal to $-I_{n+1}$.
It is obvious that the matrix $A_1, A, D$, and $E$ are symmetric. Thus, we have to show that

$$
A_1 A_1^T=\left(\begin{array}{c|c}
D & B \\
\hline
B^T& A+E\\
\end{array}\right)\left(\begin{array}{c|c}
D^T & B \\
\hline
B^T& A^T+E^T\\
\end{array}\right)= -I_{n+2}.$$
In other words, we have to show that following three identities are hold :

\begin{align}
D ^2 +B B^T&=-I_2, \label{eq:1}\\
D B +B(A+E)&= O_{2\times n},  \label{eq:2}\\
B^T B + (A+E)^2 &=-I_n.  \label{eq:3}
\end{align}

Firstly, we verify the equality of (\ref{eq:1}).  By the assumptions, we have that $A^2 = -I_n$, $\alpha^2+\beta^2=-1$ and $\x_n \x_n^T = k$. Since $(\x_n ,\y_n)$ is a codeword of a self-dual code $\CC$, it is also clear that $\x_n \x_n^T + \y_n \y_n^T =0$ and $(\x_n, \y_n ) G^T = \x_n +\y_n A = O_n$. Thus, $\y_n \y_n^T =-k$, $\y_n A= -\x_n$ and $\x_n A = \y_n$ . By direct computations, we obtain that

\begin{align*}
B AB^T&= \left( \begin{matrix}
	\alpha \x_nA + \beta \y_nA \\
	\beta \x_nA - \alpha \y_n A					
\end{matrix}\right) \left( \begin{matrix}
	\alpha \x_n + \beta \y_n \\
	\beta \x_n - \alpha \y_n 					
\end{matrix}\right)^T\\
&=\left( \begin{matrix}
	\alpha \y_n - \beta \x_n \\
	\beta \y_n + \alpha \x_n					
\end{matrix}\right) \left( \begin{matrix}
	\alpha \x_n^T + \beta \y_n^T &
	\beta \x_n^T - \alpha \y_n^T 					
\end{matrix}\right)\\
&=\left( \begin{matrix}
	-2k\alpha\beta & k(\alpha^2-\beta^2) \\
	k(\alpha^2-\beta^2) & 2k\alpha\beta
\end{matrix}\right),\\
\end{align*}
and
\begin{align*}
B EB^T&=\frac{1}{k}  \left( \begin{matrix}
	\alpha \x_n + \beta \y_n \\
	\beta \x_n - \alpha \y_n 					
\end{matrix}\right) ( s \x_n^T \x_n  + t \y_n^T \y_n -  \x_n^T \y_n - \y_n^T \x_n) \left( \begin{matrix}
	\alpha \x_n + \beta \y_n \\
	\beta \x_n - \alpha \y_n 					
\end{matrix}\right)^T\\
&=\frac{1}{k} \left( \begin{matrix}
	k\alpha s  \x_n -k\alpha \y_n -k\beta t \y_n +k \beta  \x_n  \\
	k \beta s  \x_n -k\beta \y_n + k \alpha t \y_n - k \alpha \x_n 					
\end{matrix}\right) \left( \begin{matrix}
	\alpha \x_n^T + \beta \y_n^T &
	\beta \x_n^T - \alpha \y_n^T 					
\end{matrix}\right)\\
&=\left( \begin{matrix}
	k\alpha^2 s + 2 k\alpha\beta + k\beta^2 t  & k\alpha \beta s +k\beta^2 - k \alpha^2 - k \alpha \beta t  \\
	k\alpha \beta s +k\beta a^2 - k \alpha^2 - k \alpha \beta t  &k\alpha^2 t  - 2 k\alpha\beta +  k\beta^2 s
\end{matrix}\right).\\
\end{align*}
Therefore,
$$B(A+E)B^T =BAB^T+BEB^T =\left(\begin{matrix}
k\alpha^2 s +  k\beta^2 t  & k\alpha \beta s - k \alpha \beta t  \\
k\alpha \beta s  - k \alpha \beta t  &k\alpha^2 t   +  k\beta^2 s
\end{matrix} \right),$$
and
\begin{align*}
D&= -\frac{1}{k^2} (B(A+E)B^T) B B^T \\&= -\frac{1}{k}  \left(\begin{matrix}
	\alpha^2 s +  \beta^2 t  & \alpha \beta s -  \alpha \beta t  \\
	\alpha \beta s  - \alpha \beta t  &\alpha^2 t   +  \beta^2 s
\end{matrix} \right)BB^T.
\end{align*}
Since
$B B^T= k \left( \begin{matrix}
\alpha^2 - \beta^2  & 2 \alpha \beta \\
2\alpha \beta & -\alpha^2 + \beta^2   					
\end{matrix}\right),$ we obtain
\begin{align*}
D&=  -\left(\begin{matrix}
	\alpha^2 s +  \beta^2 t  & \alpha \beta s -  \alpha \beta t  \\
	\alpha \beta s  - \alpha \beta t  &\alpha^2 t   +  \beta^2 s
\end{matrix} \right) \left( \begin{matrix}
	\alpha^2 - \beta^2  & 2 \alpha \beta \\
	2\alpha \beta & -\alpha^2 + \beta^2   					
\end{matrix}\right) \\ &=\left(\begin{matrix}
	\alpha^2 s -  \beta^2 t  & \alpha \beta (s +t) \\
	\alpha \beta (s +t)   &-\alpha^2 t   +  \beta^2 s
\end{matrix} \right).
\end{align*}
Hence,
\begin{align*}
D^2+ BB^T&=\left(\begin{matrix}
	\alpha^2 s -  \beta^2 t  &- \alpha \beta (s +t) \\
	- \alpha \beta (s +t)   &\alpha^2 t   -  \beta^2 s
\end{matrix} \right)^2 +  k \left( \begin{matrix}
	\alpha^2 - \beta^2  & 2 \alpha \beta \\
	2\alpha \beta & -\alpha^2 + \beta^2   					
\end{matrix}\right)\\
&=\left(\begin{matrix}
	-\alpha^2 s^2 -\beta^2 t^2  &-\alpha \beta ( s^2 - t^2) \\
	-\alpha \beta ( s^2 - t^2) & -\alpha^2 t^2 -\beta^2s^2
\end{matrix} \right) +  k \left( \begin{matrix}
	\alpha^2 - \beta^2  & 2 \alpha \beta \\
	2\alpha \beta & -\alpha^2 + \beta^2   					
\end{matrix}\right)\\
&=\left(\begin{matrix}
	\alpha^2 (k-s^2) -\beta^2(k+ t^2)  & \alpha \beta (2k  - s^2 + t^2) \\
	\alpha \beta (2k  - s^2 + t^2)& -\alpha^2(k+ t^2) +\beta^2(k-s^2)
\end{matrix} \right).
\end{align*}
Since $s^2 = -1+k$ and $t^2=-1-k$, we have that $k-s^2 = 1$, $k+t^2 = -1$ and $-s^2+t^2=2k$. Therefore,
\begin{align*}
D^2+ BB^T&=\left(\begin{matrix}
	\alpha^2  + \beta^2  & \alpha \beta (2k  - 2k) \\
	\alpha \beta (2k  - 2k)& \alpha^2 +\beta^2
\end{matrix} \right)\\&= -I_2,
\end{align*}
which is desired. The identities (\ref{eq:2}) and (\ref{eq:3}) are verified by similar computations.
\end{proof}

\section{Proof of Theorem \ref{SymBuildingup2}}
\begin{proof}	
It is easy to check that $A_2$ is symmetric.
Therefore, we have only to show that $A_2$ is anti-orthogonal, i.e.,
$$
\left(\begin{array}{c|c}
H & M \\
\hline
M^T& A+M^T (H-P)^{-1} M
\end{array}\right)\left(\begin{array}{c|c}
H & M \\
\hline
M^T& A+M^T (H-P)^{-1} M
\end{array}\right)=-I_{n+2}.$$
In other words, we have to show that following three identities are hold :

\begin{align}
H^2 +MM^T&=-I_2, \label{eq:4}\\
HM + M(A + M^T(H-P)^{-1}M)& = O_{2\times n},  \label{eq:5}\\
MM^T+(A + M^T(H-P)^{-1}M)^2 &=-I_n.  \label{eq:6}
\end{align}

We note that, with the assumption, $MA=PM$ and $HP=PH$ by Propositions \ref{GP1} and \ref{GP2}.

First, it is easy to show that the identity (\ref{eq:4}) is true from the equation (\ref{HP=PH}). For the identity (\ref{eq:5}), we calculate that
\begin{align*}
HM  +M(A+M^T(H-P)^{-1}M)&= HM + MA + MM^T(H-P)^{-1}M\\
&=(H + P)M + MM^T (H-P)^{-1}M\\ &=((H + P)(H-P) + MM^T)(H-P)^{-1}M\\&=O_2 (H-P)^{-1}M\\&=O_{2\times n}
\end{align*}
and the result follows.

Finally, for the identity (\ref{eq:6}), we expand the left hand side of (\ref{eq:6}):
\begin{align}
&M^TM + (A + M^T(H-P)^{-1}M)^2 \nonumber \\
\begin{split}
	&=M^TM+A^2 + AM^T(H-P)^{-1}M+ M^T(H-P)^{-1}MA \label{eq:7} \\&\quad\quad\quad\quad\quad\quad\quad\quad\quad+ M^T(H-P)^{-1}MM^T(H-P)^{-1}M.
\end{split}
\end{align}

Note that $A^2$, the second term of (\ref{eq:7}) equals $-I_n$. We compute the sum of (\ref{eq:7}) except the last term: 	\begin{align*}
& M^TM+A^2+AM^T(H-P)^{-1}M+ M^T(H-P)^{-1}MA\\
&=  M^TM-I_n+M^TP(H-P)^{-1}M+ M^T(H-P)^{-1}PM\\
&=  -I_n+M^T(I_n+P(H-P)^{-1}+ (H-P)^{-1}P)M\\
&= -I_n  + M^T(H-P)^{-1}((H-P)^2+(H-P)P+ P(H-P))(H-P)^{-1}M\\
&=  -I_n+M^T(H-P)^{-1}(H^2-P^2)(H-P)^{-1}M.
\end{align*}
And we put $MM^T=-(H+P)(H-P)$ in the last term of (\ref{eq:7}) to calculate
\begin{align*}
& M^T(H-P)^{-1}MM^T(H-P)^{-1}M\\
&=-M^T(H-P)^{-1}(H+P)(H-P)(H-P)^{-1}M\\
&=-M^T(H-P)^{-1}(H^2 -P^2)(H-P)^{-1}M.
\end{align*}
Therefore, we obtain that
\begin{align*}
&MM^T+(A + M^T(H-P)^{-1}M)^2 = -I_n,
\end{align*}
and this is desired.	\end{proof}

\end{document}